\documentclass[12pt]{article}
\usepackage{graphicx}
\usepackage{amssymb, amsmath, amsthm}

\newtheorem{thm}{Theorem}
\newtheorem{lem}[thm]{Lemma}

\usepackage{mathptmx}
\usepackage{array}
\usepackage{enumerate}

\begin{document}
\title{On the 4-adic complexity of the two-prime quaternary generator}
\author{Vladimir Edemskiy$^{1}$, Zhixiong Chen$^2$\\
1. Department of Applied Mathematics and Informatics,\\
Yaroslav-the-Wise Novgorod State University, \\
Veliky Novgorod, 173003, Russia\\
vladimir.edemsky@novsu.ru\\
2. Key Laboratory of Applied Mathematics, Putian University, \\ Putian, Fujian
351100, P. R. China\\
ptczx@126.com(Corresponding author)
}

%
\maketitle              
\begin{abstract}

R. Hofer and A. Winterhof proved that the 2-adic complexity of the two-prime (binary) generator
of period $pq$ with two odd primes $p\neq q$ is close to its period and it can
attain the maximum in many cases.

When the two-prime generator is applied to producing quaternary
sequences, we need to determine the 4-adic complexity. We present
the formulae of possible values of the 4-adic complexity, which is
larger than $pq-\log_4(pq^2)-1$ if $p<q$. So it is good enough to
resist the attack of the rational approximation algorithm.

\textbf{Keywords}. Cryptography, Feedback with carry shift registers, Two-prime generators, Quaternary sequences, 4-Adic complexity
\end{abstract}

\section{Introduction}

Feedback with carry shift registers (FCSRs), which were invented by
M. Goresky and A. Klapper in 1990's, play an important role in
spread-spectrum multiple-access communication and cryptography for
the design of pseudo-random sequences. FCSRs can be implemented in
hardware for high speed and they have an algebraic theory parallel
to that of linear feedback shift registers(LFSRs). We refer the
reader to the monograph
\cite{GK2012} for the theory on FCSRs.

Let $m\geq 2$ and $s^\infty = (s_{0}, s_{1},\ldots, s_{T-1})$ be a $T$-periodic $m$-ary sequence over $\mathbb{Z}_m$, the integer residue ring modulo $m$.
The shortest length (denoted by $\Phi_m(s^\infty)$) of an FCSR that can generate $s^\infty$ is called the \emph{$m$-adic complexity}, which is one of the most important cryptographic measures.
We have
\begin{equation}\label{Phi-m}
 \Phi_m(s^\infty)= \left\lceil \log_m\left (\frac{m^T-1}{\gcd(S(m), m^T-1)}\right ) \right\rceil,
\end{equation}
where $S(X)=\sum_{i=0}^{T-1} s_i X^i \in \mathbb{Z}[x]$ and $\lceil z \rceil$ is the smallest integer that
is equal to or greater than $z$.
It is excepted that  the $m$-adic complexity of an $m$-ary sequence is as large as possible.
Any $m$-ary sequence with small $m$-adic complexity is not suitable for cryptographic applications. Clearly if $\gcd(S(m), m^T-1)=1$, the $m$-adic complexity achieves the maximum.

In most references, the $2$-adic complexity of binary sequences (i.e., when $m=2$) were discussed. In particular,  a new way
was developed by H. Xiong, L. Qu and C. Li \cite{XQL} to determine the 2-adic complexity by computing Gauss period. It is very efficient for binary sequences
with optimal autocorrelation including Legendre sequence, Hall's sextic residue sequences and two-prime sequences, see \cite{XQL,H}.
After that, the 2-adic complexity of (generalized) cyclotomic sequences derived from (generalized) cyclotomic residue classes in $\mathbb{Z}_N$, the integer residue ring modulo $N$,
has been highly paid attention by researchers, see \cite{DZN,HW,H,KK,SWY,SYC,SWY1,XZS,XQL,YYL,YF,YYSS,ZZYF,ZSY}.

For the case $m>2$, one can find many observations about the $m$-adic complexity of $m$-ary sequences,
see \cite{GK2012,JXYF,K2007,XK,ZL2012}. However, it seems that the study of $m$-adic complexity of $m$-ary sequences has not been fully developed.

In particular, from an engineering standpoint, the preferred
alphabet sizes are $m=2,4,8,\ldots$, because of the compatibility
with the binary $\{0, 1\}$ nature of data representation in
electronic hardware. So it is of interest to consider the case of
$m=4$, i.e., quaternary sequences, see \cite{Kumar}.  Very recently,
S. Qiang et al consider the 4-adic complexity of the quaternary
cyclotomic sequences with
period $2p$ \cite{QLYF}. Partially motivated by \cite{QLYF}, we will discuss the 4-adic complexity of the two-prime quaternary sequences investigated in our earlier work \cite{C,E}.\\

Let $p$ and $q$ be two distinct odd primes with $\gcd(p-1,q-1)=4$ and
$e=(p-1)(q-1)/4$. By the Chinese Remainder Theorem there exists a
common primitive root $g$ of both $p$ and $q$. There also exists an
integer $h$ satisfying
$$
h\equiv g\pmod{p},\,\  h\equiv 1\pmod{q}.
$$
Below we always fix the definitions of $g$ and $h$.  Since $g$ is a
primitive root of both $p$ and $q$, by the Chinese Remainder Theorem
again, the multiplicative order of $g$
modulo $pq$ is $e$.

Define the generalized cyclotomic classes of order $4$ modulo $pq$ as
$$
D_i=\{g^sh^i \pmod {pq} : s=0,1,\ldots,e-1\},~~ 0\le i<4,
$$
and we have
$$
\mathbb{Z}_{pq}^*=D_0\cup D_1 \cup D_2 \cup D_3.
$$
We see that each $D_i$ has the cardinality  $ \frac{(p-1)(q-1)}{4}$  for $0\leq i<4$.
We note that $h^4\in D_0$, since otherwise, we write $h^4\equiv g^sh^i \pmod {pq}$ for some $0\le s<e$ and $1\le i<4$ and
get $g^{e-s}h^{4-i}=1 \in D_0$, a contradiction.

We also define
$$
P=\{p,2p,\ldots,(q-1)p\}, ~ Q=\{q,2q,\ldots,(p-1)q\}, ~R=\{0\}.
$$
Then we define the quaternary sequence $e^\infty = (e_{0}, e_{1},\ldots, e_{pq-1})$ over $\mathbb{Z}_4=\{0,1,2,3\}$ of length $pq$ by
\begin{equation}\label{quaternary}
e_u=\left\{
\begin{array}{ll}
2, & \mathrm{if}\,\ u\pmod {pq}\in Q\cup R,\\
0, & \mathrm{if}\,\ u\pmod {pq}\in P,\\
i, & \mathrm{if}\,\ u\pmod {pq}\in D_i, i=0,1,2,3.
\end{array}
\right.
\end{equation}
The linear complexity and trace representation of $e^\infty$ have been considered in \cite{C,E}.
Here we will determine the $4$-adic complexity of $e^{\infty}$.

\section{$4$-Adic complexity of $e^{\infty}$ : main contribution}

 Our main contribution is described in the following theorem.

\begin{thm}\label{t1}
Let $p$ and $q$ be two distinct odd primes with $\gcd(p-1,q-1)=4$.
Let $e^\infty$ be a quaternary sequence of length $pq$ defined in
(\ref{quaternary}). Then  the 4-adic complexity $\Phi_4(e^\infty)$
of $e^\infty$ satisfies
$$
\Phi_4(e^\infty)\in \left\{ \left\lceil \log_4\left(\frac{4^{pq}-1}{\max(r_1, r_2)}\right) \right\rceil,  \left\lceil \log_4\left(\frac{4^{pq}-1}{(2pq+1)\cdot \max(r_1, r_2)}\right) \right\rceil \right\}
$$
if $p\equiv q+4 \pmod{8}$; and
$$
\Phi_4(e^\infty)\in\left\{ \left\lceil \log_4\left(\frac{4^{pq}-1}{\max(r_1, r_2)}\right)\right\rceil, \left\lceil \log_4\left(\frac{4^{pq}-1}{(6pq+1)\cdot \max(r_1, r_2)}\right)\right\rceil \right\}
$$
if $p\equiv q\equiv 5 \pmod{8}$.

In the equations above, $r_1= \gcd \left (p+3,4^q-1\right )$ and
$$
r_2= \begin{cases}
\gcd \left (q-1,4^p-1\right ),&\text{ if } q\not \equiv 1 \pmod{3},\\
\gcd \left (q-1,4^p-1\right )/3,&\text{ if } q \equiv 1 \pmod{3}.
\end{cases}
$$
\end{thm}

We remark that either $r_1=1$ or $r_2=1$. In fact, let $r_1>1$ and $r_2>1$.
 Assume that $\overline{r}_1$ is a prime divisor of $r_1$, then $4^q \equiv 1 \pmod{\overline{r}_1}$. Hence $q|(\overline{r}_1-1)$ and $\overline{r}_1\geq 1+2q$.
 On the other hand, since $\overline{r}_1|(p+3)$, it follows
that $1+2q\leq \overline{r}_1< p+3$ and hence $p> 2q-2$. Similarly,
assume that $\overline{r}_2$ is a prime divisor of $r_2$, we will
get $1+2p\leq \overline{r}_2< q-1$ and hence $p< (q-2)/2$. A
contradiction. We illustrate Theorem \ref{t1} by some examples in
Table \ref{table-ex}.

If we suppose $p<q$. We find that $\gcd(q-1,4^p-1)\leq (q-1)/4$ and
$$
(6pq+1)\max(r_1,r_2)\leq (6pq+1)(q-1)/4<3pq^2/2.
$$
Then by (1) we obtain that $ \Phi_4(e^\infty)>pq -\log_4(pq^2)-1$.

\begin{table}
\centering
 \caption{Examples}\label{table-ex}
    \begin{tabular}
{|c|c|c|c|c|} \hline $p$ & $q$ & $r_1$ & $r_2$ & $\Phi_4(e^\infty)$  \\
\hline
 41, 173, 349, 569 & 5 & 11 & 1 & $pq-1$ \\
\hline
 617, 1237, 1609 & 5 & 31 & 1  & $pq-2$  \\
\hline
 1361 & 5 & 341& 1 & $pq-4$   \\
\hline
 233 & 29 & 59 & 1 & $pq-2$    \\
\hline
   5 & 89, 353, 397, 617, 1321 & 1 & 11 & $pq-1$  \\
\hline
    5 & 1117,  1489, 1613 & 1 & 31 & $pq-2$   \\
\hline
     5 & 2729 & 1 & 341  &  $pq-4$  \\
\hline
\end{tabular}
\end{table}

We will use the Hall polynomial for the studying of  4-adic
complexity of above mentioned sequences. Their properties will be
studied  in the next section.

\section{Auxiliary results and proof of main contribution}

To prove Theorem \ref{t1}, we only need to determine the value of
$\gcd(E(4), 4^{pq}-1)$ by  (\ref{Phi-m}), where
$$
E(X)=\sum\limits_{0\leq u<pq} e_uX^u\in \mathbb{Z}[X], ~~~ e_u
\text{ is in  (\ref{quaternary})}.
$$
Here and hereafter, the polynomials are with integer coefficients.

Since $\gcd(4^p-1, 4^q-1)=3$\footnote{Let $\ell=\gcd \left
(4^p-1, 4^{q}-1 \right)$. Then there
 exist $k, l$ such that $kp+lq=1$ and $4^{kp+lq}\equiv 1\pmod {\ell}$, i.e., $4\equiv 1\pmod {\ell}$ and $\ell=3$.} for two distinct odd primes $p$ and $q$, one can write
$$
4^{pq}-1=(4^{p}-1)\cdot \frac{4^{q}-1}{3}\cdot \frac{3(4^{pq}-1)}{(4^{p}-1)(4^{q}-1)},
$$
from which we turn to consider
$$
 \gcd(E(4), 4^{p}-1), ~~~  \gcd\left(E(4), \frac{4^{q}-1}{3}\right) ~~~ \text{ and }  \gcd\left(E(4), \frac{3(4^{pq}-1)}{(4^{p}-1)(4^{q}-1)}\right).
$$

Then we will use the following lemma to finish the proof of Theorem \ref{t1}.

\begin{lem}\label{gcd}
 Let $p$ and $q$ be two distinct odd primes with $\gcd(p-1,q-1)=4$.

(i). We have $\gcd\left(E(4), \frac{4^{q}-1}{3}\right)=\gcd(4^q-1, p+3)$.

(ii).  We have
$$
\gcd(E(4), 4^{p}-1)= \begin{cases}
\gcd \left (q-1,4^p-1\right ),&\text{ if } q\not \equiv 1 \pmod{3},\\
\gcd \left (q-1,4^p-1\right )/3,&\text{ if } q \equiv 1 \pmod{3}.
\end{cases}
$$
(iii). Let $d=\gcd\left(E(4), \frac{3(4^{pq}-1)}{(4^{p}-1)(4^{q}-1)}\right)$. If $d>1$, then $d$ is a
prime and
$$
d=
\begin{cases}
2pq+1, &\text{ if } p\equiv q+4\pmod{8},\\
6pq+1, &\text{ if } q\equiv p \equiv 5\pmod{8}.
\end{cases}
$$
\end{lem}

For our purpose, we need to consider the \emph{Hall polynomial} $H_j(X)\in \mathbb{Z}[X]$ of $D_j$, the generalized cyclotomic classes modulo $pq$ defined in Sect.1,
where
$$
H_j(X)=\sum\limits_{u\in D_j} X^u, ~~~ j=0,1,2,3.
$$
It is clear for $0\leq j,k <4$
\begin{equation}\label{eq2}
H_j(X^u)\equiv H_{j+k}(X) \pmod{X^{pq}-1} \text{ for } u\in D_k,
\end{equation}
where $j+k$ in the subscript of $H$ is taken modulo 4.

Let
$$
U(X)=H_1(X)+2H_2(X)+3H_3(X).
$$
Then we have
$$
E(X)=U(X)+2\sum\limits_{0\leq u<p} X^{uq},
$$
and hence
\begin{equation}
\label{poly}
E(4)=U(4)+2\sum\limits_{0\leq u<p} 4^{uq}=U(4)+\frac{2(4^{pq}-1)}{4^{q}-1}.
\end{equation}

So we will consider $U(4)$. The value of $U(4^{h^2})$ modulo
$(4^{pq}-1)$ helps us to get some intermediate results, where $h$ is
defined in Sect.1.

\subsection{Product of $U(4)$ and $U(4^{h^2})$ modulo $4^{pq}-1$}

According to  (\ref{eq2}), we have since $h^2\in D_2$
$$
U(4^{h^2})\equiv H_3(4)+2H_0(4)+3H_1(4)\pmod{4^{pq}-1}.
$$

To compute $U(4)\cdot U(4^{h^2})$ modulo $4^{pq}-1$, we only need to study the products
$$
 H_l(4)H_{l+k}(4)= \sum _{x \in
D_l, y\in D_{l+k}}\ 4 ^{x+y}, ~~~ 0\leq l,k <4.
$$

\begin{lem} \label{l1}
 Let $p$ and $q$ be two distinct odd primes with $\gcd(p-1,q-1)=4$.

(i). The number of solutions $(x,y)$ of congruence
$$
x+y\equiv 0 \pmod {pq}, ~~~ x \in D_l, ~ y \in D_{l+k}
$$
is given by $\frac{(p-1)(q-1)}{4}$ for all $0\leq l<4$ if $k=0$ and
$ \frac{(p-1)(q-1)}{16}$ is odd or $ k=2 $ and
$\frac{(p-1)(q-1)}{16}$ is even, and otherwise  the number of
solutions of this congruence is equal to 0.

(ii). For each fixed $u \in \{1, 2, \dots, q-1\}$, the number of solutions $(x,y)$
of congruence
$$
\left\{ \begin{array}{l}
x+y\equiv u p \pmod {pq},  \\
x+y\not\equiv 0 \pmod {q},
\end{array} \right. ~~~ x \in D_l, ~ y \in D_{l+k}
$$
is given by $\frac{(p-1)(q-5)}{16}$ for all $0\leq l<4$ if $k=0$ and $ \frac{(p-1)(q-1)}{16}$ is odd or $ k=2 $ and
$\frac{(p-1)(q-1)}{16}$ is even, and otherwise  the number of solutions of
this congruence is equal to $\frac{(p-1)(q-1)}{16}$.

(iii). For each fixed $v \in \{1, 2, \dots, p-1\}$, the  number of solutions
$(x,y)$ of congruence
$$
\left\{ \begin{array}{l}
x+y\not\equiv 0 \pmod {p},  \\
x+y \equiv v q \pmod {pq},
\end{array} \right. ~~~ x \in D_l, ~ y \in D_{l+k}
$$
is given by $\frac{(p-5)(q-1)}{16}$ for all $0\leq l<4$ if $k=0$ and
$ \frac{(p-1)(q-1)}{16}$ is odd or $ k=2 $ and
$\frac{(p-1)(q-1)}{16}$ is even, and otherwise  the number of
solutions of this congruence is equal to $\frac{(p-1)(q-1)}{16}$.
\end{lem}

\begin{proof}
According to \cite[Lemma 3.3]{HYW}, we have $-1\in D_0$ when $
(p-1)(q-1)/16$ is odd, and $-1\in D_2$ when $(p-1)(q-1)/16$ is even.

(i). For $x \in D_l, y \in D_{l+k}$, we see that $x+y\equiv 0 \pmod {pq}$ iff
$y\equiv -x \pmod {pq}$ iff $-1\in D_k$. So when $-1\in D_0\cup D_2$, $(x,-x)$ are solutions of
$x+y\equiv 0 \pmod {pq}$ for each $x\in D_l$, then the number of solutions is $|D_l|=\frac{(p-1)(q-1)}{4}$.

(ii). If $-1\in D_0$, then the number of
solutions of  congruences $x+y\equiv u p \pmod {pq}, x+y\not\equiv
0 \pmod {q}$ is equal to the number of solutions of $x-y\equiv u p
\pmod {pq}$.  Thus, the statement of this lemma for $k=0$ follows
from \cite[Lemma 4]{W} and for $k\neq 0$ from \cite[Lemma 2]{W}, respectively.

If $-1\in D_2$, then the number of
solutions of congruences $x+y\equiv u p \pmod {pq}$, $x+y\not\equiv 0 \pmod {q}$ is equal to the number
of solutions of $x-z\equiv u p \pmod {pq}$ for $x \in D_l,  z \in
D_{l+k+2}$. So similarly, the statement of this lemma for $k=2$
follows from \cite[Lemma 4]{W} and from \cite[Lemma 2]{W} for otherwise.

(iii). We can prove (iii) in the same way as in (ii).
\end{proof}

The number of solutions $(x,y)$ of congruence
$$
1+x\equiv y \pmod {pq}, ~~~ x \in D_{i}, ~ y \in D_{j}
$$
has been deeply studied in the literature, see e.g. \cite{W}. The
symbol $(i,j)_{4}$, called the cyclotomic numbers of order  four
modulo $pq$, is used to denote the number of solutions, that is,
$$
(i,j)_4=|(1+D_i)\cap D_j|, ~~~ 0\leq i,j,<4.
$$

\begin{lem}(\cite{W})
 \label{l2}
 Let $p$ and $q$ be two distinct odd primes with $\gcd(p-1, q-1)=4$. Let $M=\frac{(p-2)(q-2)-1}{4}$.

Then the cyclotomic numbers are given in Table \ref{table-even} for
all $0\leq i,j<4$ if $\frac{(p-1)(q-1)}{16}$ is even, and in Table
\ref{table-odd} otherwise, where $a, b \ : pq=a^2+4b^2$ with
$a\equiv 1\pmod{4}$ and
$$
\left\{ \begin{array}{l}
A=(-a+2M+3)/8,  \\
B=(-a-4b+2M-1)/8, \\
C=(3a+2M-1)/8, \\
D=(-a+4b+2M-1)/8, \\
E=(a+2M+1)/8,\\
F=(3a+2M+5)/8, \\
G=(-a+4b+2M+1)/8, \\
H=(-a+2M+1)/8, \\
I=(-a-4b+2M+1)/8, \\
J=(a+2M-1)/8.
\end{array}
\right.
$$

\begin{table}
\centering
 \caption{$(p-1)(q-1)/16$  is even}\label{table-even}
    \begin{tabular}
{|c|c|c|c|c|} \hline $( i,j)_4$ & 0 & 1 & 2 & 3  \\
\hline
 0 & $A$ & $B$ &
$C$ & $D$
\\ \hline
 1 & $E$ & $E$ & $D$ & $B$ \\ \hline
  2 & $A$ & $E$ & $A$ & $E$  \\ \hline
  3 & $E$ & $D$ & $B$ & $E$ \\ \hline
\end{tabular}
\end{table}

\begin{table}
\centering
  \caption{$(p-1)(q-1)/16$  is odd}\label{table-odd}
\begin{tabular}
{|c|c|c|c|c|} \hline
 $(i,j)_4$ & 0 & 1 & 2 & 3  \\ \hline
0 & $F$ & $G$ & $H$ & $I$ \\ \hline
 1 & $G$ & $I$ & $J$ & $J$
\\ \hline
2 & $H$ & $J$ & $H$ & $J$  \\ \hline 3 & $I$ & $J$ & $J$ & $G$ \\
\hline
\end{tabular}
\end{table}
\end{lem}

\begin{lem} \label{l3}
 Let $p$ and $q$ be two distinct odd primes with $\gcd(p-1, q-1)=4$.
Let $d>1$  be a divisor of $\frac{3(4^{pq}-1)}{(4^{p}-1)(4^{q}-1)}$. For $0\leq l,k<4$, we have
$$
 H_l(4)\cdot H_{l+k} (4 )\equiv \sum _{f=0}^{3}( k,f)_{4} H_{f+l} (4) +\Delta
 \pmod{d},
$$
where $\Delta =\frac{(p+1)(q+1)-4}{8}$ if $k=0$ and
$\frac{(p-1)(q-1)}{16}$ is odd or $k=2$ and $\frac{(p-1)(q-1)}{16}$
is even, and otherwise $\Delta =-\frac{(p-1)(q-1)}{8}$.
\end{lem}

\begin{proof}
By definition we have
$$
H_l(4)\cdot H_{l+k} (4 ) = \sum _{x \in D_l, y\in D_{l+k}}\ 4
^{x+y}.
$$

Firstly, we compute
$$
\sum _{\substack{x \in D_l, y\in D_{l+k},\\ x+y \not \equiv
0\pmod{q},\\x+y \not\equiv 0\pmod{p}}}\ 4 ^{x+y} \equiv \sum _{\substack{x \in D_l, y\in D_{l+k},\\ x+y \not \equiv
0\pmod{q},\\x+y \not\equiv 0\pmod{p}}}\ 4 ^{x(yx^{-1}+1)}\pmod{4^{pq}-1},
$$
 where $x^{-1}$ is an inverse of $x$ modulo $pq$. Since
$yx^{-1}\pmod{pq} \in D_{k}$, $x+y \not\equiv 0\pmod{p}$ and $x+y
\not\equiv 0\pmod{q}$, it follows that
$$
 \sum _{\substack{x \in D_l, y\in D_{l+k},\\ x+y \not \equiv
0\pmod{q},\\x+y \not\equiv 0\pmod{p}}}\ 4 ^{x(yx^{-1}+1)} \equiv
\sum _{\substack{x \in D_l, w\in D_{k}\\ w+1 \in
\mathbb{Z}^*_{pq}}}\ 4 ^{x(w+1)}
$$
$$
\equiv \sum_{f=0}^3 \ \sum_{
\substack{w\in D_{k},\\
w+1\in D_{f}}} \ \sum _{x \in D_l}\ 4 ^{x(w+1)}\equiv \sum_{f=0}^3 \ \sum_{
\substack{w\in D_{k},\\
w+1\in D_{f}}} H_{f+l}(4)\pmod{4^{pq}-1}.
$$
The last `$\equiv$' comes from  (\ref{eq2}). Now since
$|(D_{k}+1)\cap D_f|=(k,f)_4$, we obtain
$$
\sum_{\substack{x \in D_l, y\in D_{l+k},\\ x+y \not \equiv
0\pmod{q},\\x+y \not\equiv 0\pmod{p}}}\ 4 ^{x+y} \equiv \sum_{f=0}^3
(k,f)_4 H_{f+l} ( 4)\pmod{4^{pq}-1}.
$$

Secondly,  we calculate
$$
\Gamma=\sum_{\substack{x \in D_l, y\in D_{l+k},\\ x+y\equiv 0\pmod{pq}}}
4^{x+y}+\sum_{\substack{x \in D_l, y\in D_{l+k},\\ x+y\equiv
0\pmod{p},\\x+y \not\equiv 0\pmod{q}}} 4^{x+y}+\sum_{\substack{x \in
D_l, y\in D_{l+k},\\ x+y\equiv 0\pmod{q},\\x+y \not\equiv
0\pmod{p}}} 4^{x+y}.
$$

 We consider the case when $k=0$ and $
(p-1)(q-1)/16$ is odd or $ k=2 $ and $(p-1)(q-1)/16$ is even. In
this case, by Lemma \ref{l1} (i) we get
$$
\sum_{\substack{x \in D_l, y\in D_{l+k},\\ x+y\equiv 0\pmod{pq}}}
4^{x+y}\equiv \frac{(p-1)(q-1)}{4} \pmod{4^{pq}-1}.
$$
By Lemma \ref{l1} (ii), we get
$$
\sum_{\substack{x \in D_l, y\in D_{l+k},\\ x+y\equiv 0\pmod{p},\\x+y
\not\equiv 0\pmod{q}}} 4^{x+y}
\equiv
\sum_{\ell=1}^{q-1}
 \sum_{\substack{x \in D_l, y\in D_{l+k},\\ x+y\equiv \ell p\pmod{pq},\\x+y
\not\equiv 0\pmod{q}}} 4^{x+y}
$$
$$
\equiv\sum_{\ell=1}^{q-1}
\frac{(p-1)(q-5)}{16} 4^{\ell p}
\equiv \frac{(p-1)(q-5)}{16}
\cdot \left(\frac{4^{pq}-1}{4^p-1}-1\right) \pmod{4^{pq}-1}.
$$
Similarly, by Lemma \ref{l1} (iii), we get
$$
\sum_{\substack{x \in D_l, y\in D_{l+k},\\ x+y\equiv 0\pmod{q},\\x+y
\not\equiv 0\pmod{p}}} 4^{x+y}\equiv \frac{(q-1)(p-5)}{16}
\cdot \left(\frac{4^{pq}-1}{4^q-1}-1\right)  \pmod{4^{pq}-1}.
$$
Then, we have
$$
\Gamma \equiv \frac{(p-1)(q-1)}{4} - \frac{(p-1)(q-5)}{16} - \frac{(q-1)(p-5)}{16}=\Delta \pmod{d}.
$$

Putting everything together, we prove the desired result. For other cases, it can be done
in the same way.
\end{proof}

Now we prove a result of the product of $U(4)$ and $U(4^{h^2})$ modulo $4^{pq}-1$.

\begin{lem}
\label{l4} Let $p$ and $q$ be two distinct odd primes with
$\gcd(p-1, q-1)=4$.  Let $d>1$ be a divisor of
$\frac{3(4^{pq}-1)}{(4^{p}-1)(4^{q}-1)}$. We have
$$
4U(4)U(4^{h^2})\equiv\begin{cases} -2(4b+3)\mathcal{H}+5pq+9,&\text{ if }  \frac{(p-1)(q-1)}{16} \text{ is even},\\
-2(4b+3)\mathcal{H}-3pq+9,&\text{ if } \frac{(p-1)(q-1)}{16} \text{ is odd},
\end{cases} \pmod{d},
$$
where $\mathcal{H}=H_0(4)+H_2(4)-H_1(4)-H_3(4)$ and $b\in \mathbb{Z}$ satisfies $pq=a^2+4b^2$ with $a\equiv  1\pmod{4}$.
\end{lem}
\begin{proof}
Since $U(4)=H_1(4)+2H_2(4)+3H_3(4)$ and $U(4^{h^2})\equiv H_3(4)+2H_0(4)+3H_1(4)\pmod{4^{pq}-1}$,
we can write
$$
U(4)\cdot U(4^{h^2})\equiv \sum_{0\leq i,j<4} ~ a_{ij} H_i(4)\cdot H_j(4) \pmod{4^{pq}-1}
$$
for some (unique) integers $a_{ij}$ as coefficients.

Let  $(p-1)(q-1)/16$ is even. By Lemma \ref{l3}, we re-write
$$
U(4)\cdot U(4^{h^2})\equiv   L_0 \big(H_0(4)+H_2(4)\big)+L_1
\big(H_1(4)+H_3(4)\big)+L_3\pmod{d},
$$
where
$$
\left\{
\begin{array}{l}
L_0=3(0,1)_4+3(0,3)_4+2(1,0)_4+2(1,2)_4+6(1,3)_4+4(2,0)_4+\\
\hspace{180pt} 10(2,3)_4+6(3,0)_4,\\
L_1=3(0,0)_4+3(0,2)_4+6(1,0)_4+2(1,1)_4+2(1,3)_4+10(2,0)_4+\\
\hspace{180pt} 4(2,1)_4+6(3,1)_4,\\
L_3=\big( -2pq+9p+9q-16\big )/2.
\end{array}
\right.
$$

 After some tedious computations of $L_0,L_1$ by Lemma \ref{l2} (with $M=\frac{(p-2)(q-2)-1}{4}$), we derive
$$
\begin{array}{rl}
U(4)\cdot U(4^{h^2})\equiv & (-2b+9M+2)\big(H_0(4)+H_2(4)\big)\\
&+(2b+9M+5)\big(H_1(4)+H_3(4)\big)+L_3 \pmod{d},
\end{array}
$$
which leads to
$$
\begin{array}{rl}
2U(4)\cdot U(4^{h^2})\equiv &
(-4b-3)\big(H_0(4)+H_2(4)\big)+(4b+3)\big(H_1(4)+H_3(4)\big)\\
&(18M+7)(H_0(4)+H_2(4)+H_1(4)+H_3(4)) +2L_3 \pmod{d}.
\end{array}
$$
 Since $H_0(4)+H_2(4)+H_1(4)+H_3(4)\equiv 1\pmod{d}$ and
 $$
 18M+7+2L_3=18\frac{(p-2)(q-2)-1}{4}+7
+(-2pq+9p+9q-16)=\frac{5pq+9}{2},
 $$
we obtain
$$
4U(4)\cdot U(4^{h^2})\equiv
-2(4b+3)\big(H_0(4)+H_2(4)-H_1(4)-H_3(4)\big)+5pq+9 \pmod{d}.
$$

We finish the proof for the first case. For the case when $(p-1)(q-1)/16$ is odd, it can be done in the same way.
\end{proof}

The proof of the following lemma can be easily carried out by
computation as in Lemma \ref{l4}, we give this here for the
completeness.
\begin{lem}
\label{l5}
Let $p$ and $q$ be two distinct odd primes with
$\gcd(p-1, q-1)=4$.
 Let $d>1$ be a divisor of $\frac{3(4^{pq}-1)}{(4^{p}-1)(4^{q}-1)}$. For $\mathcal{H}=H_0(4)+H_2(4)-H_1(4)-H_3(4)$,
 we have
$$
\mathcal{H}^2 \equiv pq \pmod{d}.
$$
\end{lem}

\begin{proof}
Since $H_0(4)+H_2(4)+H_1(4)+H_3(4)\equiv 1\pmod{d}$,
 we obtain
\begin{multline*}
\mathcal{H}^2\equiv (2H_0(4)+2H_2(4)-1)^2 \\
\equiv
4H_0^2(4)+4H_2^2(4)+8H_0(4)H_2(4)-4H_0(4)-4H_2(4)+1\pmod{d}.
\end{multline*}

Let $(p-1)(q-1)/16$ is even. Using Lemmas \ref{l2} and \ref{l3},  we get
$$
\mathcal{H}^2\equiv M_0 H_0(4)+M_1H_1(4)+M_2
H_2(4)+M_3H_3(4)+M_5\pmod{d},
$$
where
$$
\left\{
\begin{array}{l}
M_0=4(0,0)_4+4(0,2)_4+8(2,0)_4-4,\\
M_1=4(0,1)_4+4(0,3)_4+8(2,1)_4,\\
M_2=4(0,0)_4+4(0,2)_4+8(2,2)_4-4,\\
M_1=4(0,1)_4+4(0,3)_4+8(2,3)_4,\\
 M_5=
2p+2q-4.
\end{array}
\right.
$$

 After computing $M_i$ by formulae for cyclotomic numbers, we have
 $$
 M_0=M_1=M_2=M_3=M_4=4M,
 $$
 and hence
 \begin{multline*}
\mathcal{H}^2\equiv 4M\big(H_0(4)+H_2(4)+H_1(4)+H_3(4)\big)+2p+2q-3  \\
\equiv 4M+2p+2q-3  \equiv pq \pmod{d}.
\end{multline*}

For the case when $(p-1)(q-1)/16$ is odd, it can be done in the same way.
\end{proof}

\subsection{Proof of Lemma \ref{gcd}}

It is easy to see that for any polynomial $G(X)\in \mathbb{Z}[X]$, $G(4)\equiv G(1) \pmod {3}$, so
we get
$$
\begin{array}{rcl}
 E(4)& =     & H_1(4)+2H_2(4)+3H_3(4)+2\sum_{i=0}^{p-1} 4^{iq}\\
     &\equiv & H_1(1)+2H_2(1)+3H_3(1)+2\sum_{i=0}^{p-1} 1^{iq}\\
     &\equiv & |D_1|+ 2|D_2| + 3|D_3| +2p\\
     &\equiv & 6(p-1)(q-1)/4 +2p \equiv 2p \pmod{3},
\end{array}
$$
and hence $3\nmid E(4)$.\\

Now we prove (i) and (ii).

According to \cite[Lemma 2]{C},
we see that $D_j\bmod q=
 \{1, 2, \ldots, q-1\}$  and when $s$ ranges over $\{0,
1, \ldots, e-1\}$, $g^sh^j\bmod~ q$ takes on each element of $\{1, 2,
\ldots , q- 1\}$ exactly $(p - 1)/4$ times , hence for $0\leq j<4 $ we get
$$
H_j(4)\equiv \frac{p-1}{4} \left (4+\ldots+4^{q-1}\right )
\equiv \frac{p-1}{4}\cdot  \left ( \frac{4^q-1}{4-1}-1 \right )
\pmod
{4^q-1}.
$$
Then together with $4^{iq}\equiv 1\pmod{ 4^q-1}$ for all $i\in \mathbb{N}$, we obtain
$$
E(4)\equiv 6\cdot \frac{p-1}{4}\cdot  \left ( \frac{4^q-1}{4-1}-1 \right )+ 2p
\equiv
 \frac{p+3}{2} \pmod{ 4^q-1}.
$$
Thus, due to $3\nmid E(4)$,  we derive
$$
\gcd(E(4), \frac{4^q-1}{3})= \gcd(E(4), 4^q-1)= \gcd ( 4^q-1, \frac{p+3}{2}) =\gcd ( 4^q-1, p+3).
$$

Similarly, we have
$$
H_j(4)  \equiv  \frac{q-1}{4} \left (4+\ldots+4^{p-1}\right )
        \equiv  \frac{q-1}{4}\cdot  \left ( \frac{4^p-1}{4-1}-1 \right )  \pmod{4^p-1}
$$
and
$$
E(4)\equiv -\frac{6(q-1)}{4}+ \frac{2(4^{pq}-1)}{4^q-1} \equiv -\frac{3(q-1)}{2} \pmod{4^p-1}.
$$
Thus, due to the fact\footnote{Let $3^2| (4^p-1)$. That is $4^p\equiv 1\pmod{9}$. Then $p$ divides the value of the Euler's totient function
$\varphi(9)=6$ and we have a contradiction since $p>3$.} that $3^2\nmid (4^p-1)$,  we derive
$$
\gcd (E(4),4^p-1)= \gcd ( 4^p-1, q-1)
$$
if  $q \not \equiv 1 \pmod{3}$, and otherwise
$$
\gcd(E(4), 4^p-1)= \gcd ( 4^p-1, (q-1))/3.
$$  \\

Now we turn to prove (iii).

Suppose $d>1$. Note that $d$ is odd with $3\nmid d$ since $3\nmid E(4)$. We now assume that $d_0>3$ is an odd prime with $d_0|d$.
Below we will discuss the possible value of $d_0$.

We first show $4^{p} \not\equiv 1 \pmod{d_0}$ and $4^{q} \not\equiv 1 \pmod{d_0}$ to get the possible form of $d_0$.
Since $d_0$ is a divisor of $\frac{3(4^{pq}-1)}{(4^{p}-1)(4^{q}-1)}$, we see that $d_0$ is a divisor of
$$
\frac{3(4^{pq}-1)}{4^{p}-1}=3(1+4^p+\dots+4^{(q-1)p}),
$$
from which we derive
$$
1+4^p+\dots+4^{(q-1)p} \equiv 0 \pmod{d_0}.
$$
If $4^{p} \equiv 1 \pmod{d_0}$, that is $d_0|(4^p-1)$,   we get from above
$$
0 \equiv 1+4^p+\dots+4^{(q-1)p} \equiv q \pmod{d_0}.
$$
So we have $d_0=q$ since $d_0>3$ is a prime.
Together with $d_0|E(4)$ and $d_0|(4^p-1)$,
 we derive by (ii) of this lemma,
$$
q|\gcd(E(4), 4^{p}-1) \Rightarrow q|(q-1),
$$
a contradiction. Similarly one can prove $4^{q} \not\equiv 1 \pmod{d_0}$. From discussions above, we also see that
$p\nmid d_0$ and $q\nmid d_0$.
Therefore, the congruences
\begin{equation}\label{eqnpq}
\left\{
\begin{array}{l}
4^{pq} \equiv 1 \pmod{d_0} ,\\
4^{p} \not\equiv 1 \pmod{d_0},\\
4^{q} \not\equiv 1 \pmod{d_0},\\
\end{array}
\right.
\end{equation}
tell us that $m=pq$ is the smallest integer such that $4^{m} \equiv 1 \pmod{d_0}$. Then we have $pq|(d_0-1)$ since $4^{d_0-1} \equiv 1 \pmod{d_0}$
by the Fermat Little Theorem. Then $d_0$ is of the form
$$
d_0=1+2\lambda pq, ~~ \text{for some } ~~0<\lambda\in \mathbb{Z}.
$$\\

We note that either $p\equiv q+4\pmod{8}$ or $p\equiv q\equiv 5
\pmod{8}$, since $\gcd(p-1,q-1)=4$. We have $d|U(4)$ by
 (\ref{poly}). Now we consider the following two cases.

\begin{itemize}
\item Suppose $p\equiv q+4\pmod{8}$.

\noindent In this case we see that $(p-1)(q-1)/16$ is even. Then
since $d|U(4)$,  by Lemma \ref{l4} we have
$$
-2(4b+3)\mathcal{H}+5pq+9\equiv 0 \pmod{d},
$$
that is,
$$
2(4b+3)\mathcal{H} \equiv 5pq+9 \pmod{d}.
$$
Using Lemma \ref{l5}, we obtain
\begin{equation}\label{eqn}
4(4b+3)^2pq\equiv 25p^2q^2+90pq+81 \pmod {d}.
\end{equation}

Then from  (\ref{eqn}), we obtain
\begin{equation}\label{eqn333}
 -4(4b+3)^2 +25pq+90-162\lambda \equiv 0 \pmod{d_0},
\end{equation}
where we use $1\equiv -2\lambda pq \pmod{d_0}$.

Since the left hand side of  (\ref{eqn333}) is odd, it can be
written as
\begin{equation}
\label{eq9}
-4(4b+3)^2 +25pq+90-162\lambda=\delta  (1+2\lambda pq)
\end{equation}
with odd $\delta=8\mu+3$ for some integer $\mu$.\footnote{If
$\delta$ is of the form $8\mu+1$, $8\mu+5$ or $8\mu+7$, we have from
 (\ref{eq9})$$-4(4b+3)^2 +25pq+90-162\lambda \not\equiv \delta
(1+2\lambda pq) \pmod{8}.$$} Below we consider the possible values
of $\mu$ and $\lambda$ from  (\ref{eq9}).

 1). If $\mu<0$, since $pq=a^2+4b^2$ and
 $$25pq+90-4(4b+3)^2\geq 36b^2-96b +79>0,$$
 it follows that
 $-162 \lambda <(3+8\mu)(1+2\lambda  pq)$ or
 $$162\lambda >(-3-8\mu)(1+2\lambda  pq)\geq 5+10\lambda  pq > 10\lambda  pq.$$
 We get $10pq < 162$, and hence $p=3$ and $q=5$ (or $p=5$ and $q=3$). This is a contradiction to $\gcd(p-1,q-1)=4$.

 2). If $\mu=0$, then by \eqref{eq9} we get that
\begin{equation}
\label{eqmu0}
 -4(4b+3)^2 +25pq+90-162\lambda =3(1+2\lambda  pq),
\end{equation}
 and hence $25 pq>6\lambda  pq$. So $\lambda\in \{1,2,3,4\}$.

 Let $\lambda= 1$ or $\lambda = 4$. By  (\ref{eqmu0}), we obtain $-b^2+pq\equiv  0\pmod{3}$. Since $p> 3$ and $q> 3$, it follows that $b\not \equiv 0 \pmod{3}$.
  Hence $b^2\equiv 1\pmod{3}$ and $pq\equiv b^2\equiv 1 \pmod{3}$.
 This leads to $1+2\lambda pq \equiv 0 \pmod 3$,  a contradiction to that $d_0=1+2\lambda pq$ is a prime.

 Let $\lambda=2$. We have $d_0=1+4pq$ and $d_0\equiv  5  \pmod {8}$ since $p\equiv
 q+4\pmod{8}$. By condition $d_0$ is a prime. Then there exist a primitive root $\theta$
  modulo $d_0$ and the order of $\theta$ modulo $d_0$ is equal to $\varphi(d_0)=d_0-1=4pq$, i.e.,
 $4pq$ is the smallest integer such that
$\theta^{4pq} \equiv 1 \pmod{d_0}$.
Write $2\equiv  \theta^m \pmod {d_0}$ for some integer $m$, we have
 $$
 1\equiv 4^{pq} \equiv 2^{2pq}\equiv \theta^{2mpq} \pmod{d_0},
 $$
since $d_0$ is a divisor of $4^{pq}-1$. Hence $2mpq$ is divided by $4pq$, which indicates that $m$ is
  even and 2 is a quadratic residue modulo $d_0$.
  We obtain a contradiction since 2 is a quadratic residue modulo $d_0$ iff $d_0\equiv \pm 1\pmod{8}$.

  Let $\lambda=3$. In this case we have by  (\ref{eqmu0})
  $$
  -4(4b+3)^2 +25pq+90-486=3(1+6pq),
  $$
that is, $-4(4b+3)^2 +7pq-399=0$, which leads to $-4(4b+3)^2 \equiv 0 \pmod {7}$.
So we write $4b+3=7w$ for some $ w\in \mathbb{Z}$ and get $-4\cdot7w^2+pq -57=0$. Hence $pq\equiv 1 \pmod {7}$ and
$d_0=1+6pq\equiv 0 \pmod{7}$, which contradicts to $d_0$ being prime.

3). If $\mu>0$, then by \eqref{eq9} we get that
 $$
 25pq>(3+8\mu)(1+2\lambda  pq),
 $$
from which we derive $\mu=\lambda=1$. Then $d_0=1+2pq$ (which is a prime) and $4(4b+3)^2= 3pq-83$ by \eqref{eq9} again.\\

From discussions above, we conclude that $d$ is a power of prime
$d_0$ with $d_0=1+2pq$. If we assume $d_0^2|d$, putting \eqref{eqn}
with $4(4b+3)^2= 3pq-83$ together, we get
$$
(3pq-83)pq\equiv 25p^2q^2+90pq+81 \pmod {d_0^2},
$$
that is,
$$
22p^2q^2+173pq+81 \equiv 0\pmod {4p^2q^2+4pq+1},
$$
this is  impossible.

So, we prove that $\gcd\left(E(4),
\frac{3(4^{pq}-1)}{(4^{p}-1)(4^{q}-1)}\right)$ equals either 1 or $1+2pq$, where $1+2pq$ is a
prime number.

\item Suppose $p\equiv q\equiv 5 \pmod{8}$. The proof is similar to above. We give a sketch of proof  for convenience of the reader.

In this case $(p-1)(q-1)/16$ is odd. Since $d|U(4)$ and according to Lemma \ref{l4}, we have
$$
-2(4b+3)\mathcal{H} -3pq+9 \equiv 0 \pmod {d}.
$$
Furtherly, using Lemma \ref{l5} we have
\begin{equation}
\label{eq1000}
4(4b+3)^2pq \equiv (-3pq+9)^2 \equiv 9p^2q^2-54pq+81 \pmod {d}.
\end{equation}

Since  $d_0|d$, we have from  (\ref{eq1000})
\begin{equation}
\label{eq10}
 -4(4b+3)^2 +9pq-54-162\lambda=\delta\cdot d_0=\delta  (1+2\lambda pq)
\end{equation}
with odd $\delta=8\mu+7$ for some integer $\mu$.\footnote{If
$\delta$ is of the form $8\mu+1$, $8\mu+3$ or $8\mu+5$, we have from
(\ref{eq10})
$$-4(4b+3)^2 +9pq-54-162\lambda \not\equiv \delta (1+2\lambda pq)
\pmod{8}.$$} Below we consider the possible values of $\mu$ and
$\lambda$ from  (\ref{eq10}).

 1). If $\mu\geq 0$ then by \eqref{eq10} we get $9pq>14pq$. It is  impossible.

 2). If $\mu <-1$ then by \eqref{eq10} we obtain that
 $$
 -4(4b+3)^2 +9pq-54-162\lambda \leq -9 -18 \lambda pq,
 $$
 that is,
 $$
 -4(4b+3)^2 +9pq-45 \leq 162\lambda -18 \lambda pq.
 $$

 Since $pq\geq 65$,  it follows that $162\lambda-3\lambda pq<0$,  which leads to
 $$
 -4(4b+3)^2 +9pq-45\leq -15\lambda pq\leq -15 pq.
 $$
 Then
 \begin{equation}
 \label{q1}
 24(1+4b^2)\leq 24pq \leq 4(4b+3)^2 +45.
 \end{equation}
 In this case, we see that $32b^2-96b -57\leq 0$.  According to \cite{HYW}  $b$ is
even, hence only  $b=2$ satisfies the last inequality. Then by
\eqref{q1} we get
$$
24pq \leq 4\cdot 11^2+25 =529, ~~~ \text{i.e.,} ~~~ pq<22.
$$
No such $p$ and $q$ exist.

3). If $\mu=-1$ then by \eqref{eq10} we have
\begin{equation}
\label{eq1011}
-4(4b+3)^2
+9pq-54-162\lambda =-(1+2\lambda pq),
\end{equation}
from which $b\not \equiv 0 \pmod{3}$, since otherwise
$d_0=1+2\lambda pq \equiv 0\pmod{3}$ from  (\ref{eq1011}), this
means that  $d_0$ is not a prime. So from   (\ref{eq1011}) again, we
have $1\equiv b^2\equiv 1+2\lambda pq\pmod{3}$. Then   we derive
$2\lambda pq\equiv 0\pmod{3}$ and hence $\lambda \equiv 0\pmod{3}$,
i.e., $\lambda=3, 6, 9, ...$.

It is easy to prove that the equality \eqref{eq1011} is not true for
$pq=65$.

Let $\lambda \geq 6$. In this case for $pq> 65$ we see that
$$
-4(4b+3)^2 +9pq-54= -\lambda (1 +2pq-162)\leq -6(1 +2pq-162).
$$
 Then
 \begin{equation}
\label{q2} 21(1+4b^2)\leq  21pq \leq 4(4b+3)^2 +1020.
\end{equation}
Hence $20b^2-96b -1035\leq 0$, from which we get $b= -4, -2, 2, 4$
since $b$ is even and $b\not \equiv 0\pmod{3}$. For these values of
$b$ by \eqref{q2} we get $21pq\leq 2464$ or $pq< 118$. No such $p$
and $q$ exist.

 Let $\lambda=3$. Then  $d_0=1+6 pq$ and $4(4b+3)^2=15pq - 539$ by \eqref{eq1011}.

From discussions above, we conclude that $d$ is a power of prime
$d_0$ with $d_0=1+6pq$. We can show that $(1+6pq)^2$ does not divide $d$
in the same way as before.

So, we prove that $\gcd\left(E(4),
\frac{3(4^{pq}-1)}{(4^{p}-1)(4^{q}-1)}\right)$ equals either 1 or
$1+6pq$, where $1+6pq$ is a prime number.

This completes the proof of Lemma \ref{gcd}.

\end{itemize}

\subsection{Proof of Theorem \ref{t1}.}

From
$$
\gcd\left(4^{p}-1, \frac{4^{q}-1}{3}\right)=1,
$$
we see that $\gcd(E(4), 4^{p}-1)$ and $\gcd\left(E(4), \frac{4^{q}-1}{3}\right)$ do not share the common divisor of larger than one.
That is,
any divisor larger than one of $\gcd(E(4), 4^{p}-1)$ does not divide $\gcd\left(E(4), \frac{4^{q}-1}{3}\right)$, and vice versa.

From  (\ref{eqnpq}), we see that
$$
\gcd\left(\gcd(E(4), 4^{p}-1), ~~ \gcd\left(E(4), \frac{3(4^{pq}-1)}{(4^{p}-1)(4^{q}-1)}\right)\right)=1.
$$
Since  otherwise, we assume that $d_0$ is a prime with
$d_0|\gcd(E(4), 4^{p}-1)$ and $d_0|\gcd\left(E(4),
\frac{3(4^{pq}-1)}{(4^{p}-1)(4^{q}-1)}\right)$. Then we obtain
$d_0|(4^{p}-1)$ and $d_0|(4^{pq}-1)$, which contradicts to
(\ref{eqnpq}). Similarly, we have
$$
\gcd\left(\gcd\left(E(4), \frac{4^{q}-1}{3}\right), ~~~ \gcd\left(E(4), \frac{3(4^{pq}-1)}{(4^{p}-1)(4^{q}-1)}\right)\right)=1.
$$

So we derive
$$
\begin{array}{l}
\gcd(E(4), 4^{pq}-1)=  \\
~~~~ \gcd(E(4), 4^{p}-1)\cdot\gcd\left(E(4), \frac{4^{q}-1}{3}\right)\cdot\gcd\left(E(4), \frac{3(4^{pq}-1)}{(4^{p}-1)(4^{q}-1)}\right).
\end{array}
$$
Then  Theorem \ref{t1} is proved by Lemma \ref{gcd} in terms of
\eqref{Phi-m}.

 \qed

\section{Symmetric 4-adic complexity of
$e^\infty$}

According to the discussion of H. Hu and D. Feng's in \cite{HF}, it
is interesting to determine the $2$-adic complexity of
$\tilde{s}^{\infty}$, where $\tilde{s}^{\infty}= (s_{T-1}, s_{T-2},
\ldots , s_0)$ is taken from the binary sequence $s^\infty$ of
period $T$ in inverse order. They defined
$$
\overline{\Phi}_2(s^\infty)=\min \left ( \Phi_2(s^\infty),
\Phi_2(\tilde{s}^\infty) \right ),
$$
which is called the \emph{symmetric 2-adic complexity} of
$s^\infty$. They claimed that the symmetric 2-adic complexity is
better than 2-adic complexity in measuring the security for binary
sequences.

Therefore, it is interesting to consider $\Phi_4(\tilde{e}^\infty)$
for $\tilde{e}^{\infty}= (e_{pq-1}, e_{pq-2}, \ldots , e_0)$, where
$e_u$ is defined in  (\ref{quaternary}). The generating polynomial
of $\tilde{e}^{\infty}$ is
$$
\tilde{E}(X)=\sum_{i=1}^{pq} e_{pq-i}X^{i-1}.
$$
In this case, we have the
following statement.
\begin{lem}
 \label{l7}
 Let $e^\infty$ be defined as in
\eqref{quaternary} and $\tilde{e}^\infty=(e_{pq-1},\dots,e_1,e_0)$.
Then
 \begin{enumerate} [(i)]
 \item
 $4\tilde{E}(4)\equiv  E(4)\pmod{4^{pq}-1}$  if $\frac{(p-1)(q-1)}{16}$ is odd;

\item
 $4\tilde{E}(4)\equiv   U\left (4^{h^2} \right )+2\sum_{u=0}^{p-1}
4^{uq} \pmod{4^{pq}-1}$ if $\frac{(p-1)(q-1)}{16}$ is even, where
$U(X)=H_1(X)+2H_2(X)+3H_3(X)$.
\end{enumerate}
\end{lem}
\begin{proof}By definition of $\tilde{e}^\infty$ we see that $\tilde{E}(4)=\sum_{i=1}^{pq}
e_{pq-i}4^{i-1} $. Hence
\begin{equation}
\label{eq8} 4\tilde{E}(4)=\sum_{i=1}^{pq}
e_{pq-i}4^i=\sum_{i=0}^{pq-1} e_{pq-i}4^i+e_04^{pq}-e_{pq}.
\end{equation}
It is clear that $e_{pq-i}=e_i$ for $i\in P\cup Q \cup R$.

Let $i\in \mathbb{Z}_{pq}^*$. As noted above we have  $-1\in D_0$ if
$ (p-1)(q-1)/16$ is odd and  $-1\in D_2$ if $ (p-1)(q-1)/16$ is
even, respectively. So, $e_{-i}=j$ iff $i\in D_j$ in the first case
and $e_{-i}=j$ iff $i\in D_{(j+2)\pmod{4}}$ in the second case.
Thus, this statement follows from the definition of $U(X)$,
\eqref{eq2} and  \eqref{eq8}.
\end{proof}
According to Lemmas \ref{l4} and \ref{l7},  we get the same result for
$\Phi_4(\tilde{e}^\infty)$ of $\tilde{e}^{\infty}$  as that of $e^{\infty}$ in Theorem \ref{t1}. We state below:

\begin{thm}\label{t2}
Let $p$ and $q$ be two distinct odd primes with $\gcd(p-1,q-1)=4$.
Let $e^\infty$ be a quaternary sequence of length $pq$ defined in
 (\ref{quaternary}) and $\tilde{e}^{\infty}= (e_{pq-1}, e_{pq-2}, \ldots , e_0)$.
Then we have
$$
\Phi_4(\tilde{e}^\infty) = \Phi_4(e^\infty),
$$
and hence the symmetric 4-adic complexity
$\overline{\Phi}_4(e^\infty)$ of $e^\infty$ satisfies
$$
\overline{\Phi}_4(e^\infty)=\Phi_4(\tilde{e}^\infty) = \Phi_4(e^\infty).
$$
\end{thm}

\section{Final remarks and conclusions}

Sequences generated by FCSRs share many important properties by LFSR
sequences. In this work, we continued to study quaternary sequences of period $pq$ generated by the two-prime generator.
We determined the possible values of the 4-adic complexity, which is larger than
$pq-\log_4(pq^2)$ for $p<q$.
 Our study is based
on  the properties of Hall polynomial of Whiteman generalized
cyclotomic classes. Our result showed that the 4-adic complexity of
these sequence is obviously large enough to resist against  the
Rational Approximation Algorithm for FCSR.

With the help of computer program we verified that $(2pq+1)\nmid
E(4)$ or $(6pq+1)\nmid E(4)$ for all primes $p, q$ such that $5\leq
p, q <10000$ and $pq\leq 424733$. So we conjecture that
$$
\Phi_4(e^\infty)=  \left\lceil \log_4\left(\frac{4^{pq}-1}{\max(r_1, r_2)}\right) \right\rceil,
$$
for all odd primes $p$ and $q$ with $\gcd(p-1,q-1)=4$.

\section*{Acknowledgments}

Z. Chen was partially supported by the National Natural Science
Foundation of China under grant No.~61772292, and by the Project of International Cooperation
and Exchanges NSFC-RFBR No.~61911530130.

%
%
%
 \bibliographystyle{splncs04}
%

\end{document}